\documentclass[review]{elsarticle}

\usepackage{stmaryrd}
\usepackage{amsfonts}
\usepackage{graphicx,times,amsmath}
\usepackage{graphics,color}
\usepackage{graphicx}
\usepackage{latexsym}
\usepackage{amsmath}
\usepackage{amsfonts}
\usepackage{amssymb}
\usepackage{url}
\usepackage{subfigure}
\usepackage{times}
\usepackage{color}
\usepackage{lineno}

\newtheorem{thm}{Theorem}
\newtheorem{lem}{Lemma}
\newtheorem{mydef}{Definition}
\newtheorem{ass}{Assumption}

\newtheorem{rem}{Remark}

\newtheorem{prop}{Property}

\allowdisplaybreaks

\allowbreak \allowdisplaybreaks

\begin{document}
%
%
%

\begin{frontmatter}

\title{Bridge the gap between
network-based inference method and global ranking method in personal recommendation}

\author[tongji,lab]{Xiwei Liu\corref{lxw}}
\ead{xwliu@tongji.edu.cn, xwliu.sh@gmail.com}

\cortext[lxw]{Corresponding author}

\address[tongji]{Department of Computer Science and
Technology, Tongji University, China}
\address[lab]{The Key Laboratory of Embedded System and Service Computing,
Ministry of Education, Shanghai 200092, China.}
\begin{abstract}
In this paper, we study the relationship between the network-based inference method and global ranking method in personal recommendation. By some theoretical analysis, we prove that the recommendation result under the global ranking method is the limit of applying network-based inference method with infinity times.

\end{abstract}


\end{frontmatter}

\linenumbers

\section{Introduction}
Personal recommendation \cite{lv12} has been a hot topic recently because of its wide and effective application in business, and many recommend methods have been developed. For example, global ranking method (GRM), network-based inference (NBI) method \cite{Zhou} (also called the mass-diffusion (MD) method or probabilistic spreading (ProbS) algorithm), heat-spreading (HeatS) algorithm \cite{Zhang2007,Zhou10} (also called the heat conduction (HC) or heat diffusion process). Many new algorithms have been proposed based on these methods, see \cite{Zhang07}-\cite{Zhou09} and references therein.

In this paper, we will concentrate on the relationship between NBI and GRM in the theoretical view. By studying the NBI, we propose a new algorithm by using the NBI method with multiple times. Using the matrix analysis technique, we rigorously prove that the result on GRM equals to that under the limit of NBI with infinity times.

\section{GRM and NBI: A review}
The following definitions and statement of GRM and NBI have been given in \cite{Zhou}, here we rewrite it to keep the self-integrity of this paper.
\subsection{GRM}
At first, we state the recommendation process of GRM.

{\bf GRM} sorts all the objects in the descending order of degree, and recommends those with highest degrees.

GRM lacks of personalization, but it is widely used since it is simple, the well known ``Yahoo Top 100 MTVs'', ``Amazon List of Top Sellers'', as well as the board of most downloaded articles in many scientific journals, can be all considered as results of GRM.

\subsection{NBI}
NBI sets up a bipartite to discriminate the object-set and the user-set, then uses the diffusive idea to determine the importance of node $i$ in node $j$'s sense.

At first, we present the definition of bipartite.
\begin{mydef}\label{bi}(See \cite{K2012})
A simple graph $G$ is called \emph{bipartite} if its vertex set $V$ can be partitioned into two disjoint sets $V_1$ and $V_2$ such that every edge in the graph connects a vertex in $V_1$ and a vertex in $V_2$ (so that no edge in $G$ connects either two vertices in $V_1$ or two vertices in $V_2$). When this condition holds, we call the pair $(V_1,V_2)$ be a \emph{bipartition} of the vertex set $V$ of $G$.
\end{mydef}

Consider a general bipartite network $G(V_1,V_2,E)$, $V_1=\{o_1,\cdots,o_n\}$ means the object-set, and $V_2=\{u_1,u_2,\cdots,u_m\}$ means the user-set. The $n\times m$ adjacent matrix $A$ is defined by $A=(a_{il})$, where $a_{il}=1$ if user $u_l$ has already collected object $o_i$ and $a_{il}=0$ otherwise. Moreover, assume the initial resource located in the $i$-th object is $f(o_i)\ge 0$. A sketch map of the bipartite network with three objects and four users is given in Figure \ref{f1}.

\begin{figure}
\begin{center}
\includegraphics[width=\textwidth,height=0.3\textheight]{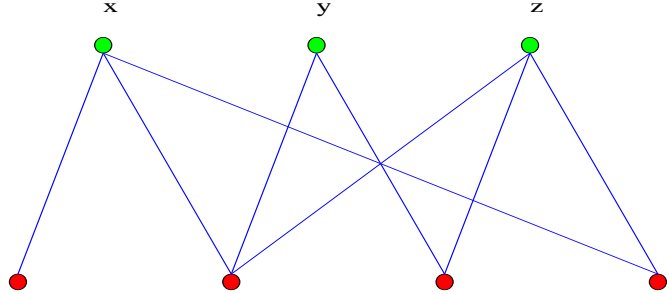}
\caption{Sketch map of the bipartite network with initial resources $x,y,z$}\label{f1}
\end{center}
\end{figure}

The first key point for NBI is that: the resource in an arbitrary $V_1$ node should be equally distributed to its neighbors in $V_2$, here we use notation $V_1\rightarrow V_2$ to denote this step. After this step, all the resource in $V_1$ flows to $V_2$, and the resource located on the $l$-th $V_2$ node is:
\begin{align}
f(u_l)=\sum_{j=1}^n\frac{a_{jl}f(o_j)}{k(o_j)}
\end{align}
where $k(o_j)$ is the degree of object $j$. Figure
\ref{f2} gives a sketch map of the diffusive result under step $V_1\rightarrow V_2$.

\begin{figure}[htp]
\begin{center}
\includegraphics[width=\textwidth,height=0.3\textheight]{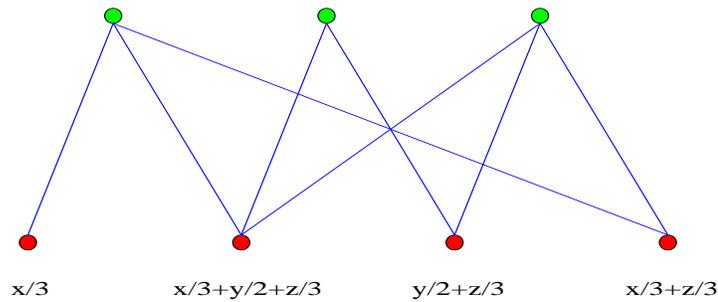}
\caption{Sketch map of the diffusive result under step $V_1\rightarrow V_2$}\label{f2}
\end{center}
\end{figure}

The second key point for NBI is that: the resource in any $V_2$ node should be equally distributed to its $V_1$ neighbors, here we use notation $V_2\rightarrow V_1$ to denote this step. After this step, all the resource in $V_2$ flows back to $V_1$, and the final resource located on $o_i$ is
\begin{align}
f^{\prime}(o_i)=\sum_{l=1}^m\frac{a_{il}f(u_l)}{k(u_l)}=
\sum_{l=1}^m\frac{a_{il}}{k(u_l)}\sum_{j=1}^n\frac{a_{jl}f(o_j)}{k(o_j)}
=\sum_{j=1}^nw_{ij}f(o_j)
\end{align}
where
\begin{align}\label{w}
w_{ij}=\frac{1}{k(o_j)}\sum_{l=1}^m\frac{a_{il}a_{jl}}{k(u_l)}
\end{align}
Therefore, if we denote $F^{(1)}=(f^{\prime}(o_1),\cdots,f^{\prime}(o_n))^T$, $W=(w_{ij})\in R^{n\times n}$ and $F=(f(o_1),\cdots,f(o_n))^T$, then
\begin{align}\label{m}
F^{(1)}=WF
\end{align}

For the example in Figure \ref{f1} and Figure \ref{f2}, after this step, see Figure \ref{f3}, the resources at these three nodes are denoted by $x^{\prime}, y^{\prime}, z^{\prime}$, which can be calculated as
\begin{align}\label{p}
\left(
\begin{array}{c}
x^{\prime}\\
y^{\prime}\\
z^{\prime}
\end{array}
\right)=\left(
\begin{array}{c}
\frac{11}{18}x+\frac{1}{6}y+\frac{5}{18}z\\
\frac{1}{9}x+\frac{5}{12}y+\frac{5}{18}z\\
\frac{5}{18}x+\frac{5}{12}y+\frac{4}{9}z
\end{array}
\right)=
\left(
\begin{array}{ccc}
\frac{11}{18}&\frac{1}{6}&\frac{5}{18}\\
\frac{1}{9}&\frac{5}{12}&\frac{5}{18}\\
\frac{5}{18}&\frac{5}{12}&\frac{4}{9}
\end{array}
\right)
\left(
\begin{array}{c}
x^{\prime}\\
y^{\prime}\\
z^{\prime}
\end{array}
\right)
\end{align}
\begin{figure}[h]
\begin{center}
\includegraphics[width=\textwidth,height=0.29\textheight]{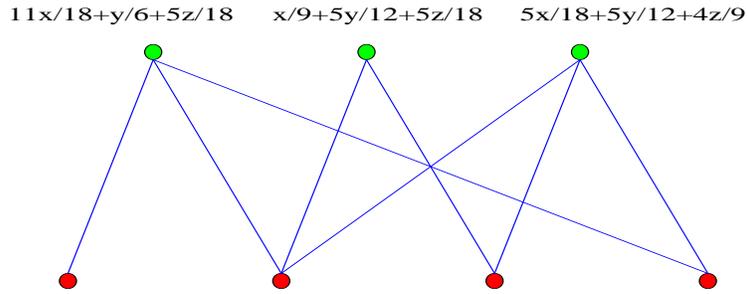}
\caption{Sketch map of the diffusive result under step $V_2\rightarrow V_1$}
\label{f3}
\end{center}
\end{figure}

{\bf NBI recommendation algorithm:} The recommendation algorithm aims at predicting user $u_l$'s personal opinions on those objects $u_l$ has not yet collected, $l=1,\cdots,m$. Set the initial resource located on each node $o_j, j=1,\cdots,n$ of $V_1$ as
\begin{align}
f(o_j)=a_{jl}
\end{align}
That is to say, if the object $o_j$ has been collected by $u_l$, then its initial resource is unit, otherwise it is zero. The initial resource can be understood as giving a unit recommending capacity to each collected object. Therefore,
\begin{align}
f^{\prime}(o_i)=\sum_{j=1}^nw_{ij}f(o_j)=\sum_{j=1}^nw_{ij}a_{jl}
\end{align}
For any user $u_i$, all his uncollected objects $o_j$ are sorted in the descending order of $f^{\prime}(o_j)$, and those objects with highest value of final resource are recommended.

\section{Some theoretical preparations}\label{model}
At the first glance, one cannot find any relationship between GRM and NBI, because GRM does not consider the personal difference, while NBI considers it; moreover, numerical examples show that NBI is better than GRM in personal recommendation \cite{Zhou}. In order to bridge the gap between them from the theoretical point, we should first present some useful lemmas and properties.

\begin{prop}
For the matrix $W=(w_{ij})$, where $w_{ij}$ is defined in (\ref{w}), it is column normalized, that is to say£¬
\begin{align}
\sum\limits_{i=1}^nw_{ij}=1, j=1,\cdots,n.
\end{align}
\end{prop}

\begin{proof}
According to the definition of $w_{ij}$, in order to prove the above property, we just need to prove that the following conclusion holds:
\begin{align}\label{n1}
\sum\limits_{i=1}^n\frac{1}{k(o_j)}\sum\limits_{l=1}^m\frac{a_{il}a_{jl}}{k(u_l)}=1. \quad~~~ \forall j=1,\cdots,n.
\end{align}

Because $k(o_j)$ is independent on parameters $i$ and $l$, so we let $k(o_j)=p(j)$, where $p(j)$ is an integer depending on $j$, which means that there are $p(j)$ nodes in $V_2$ connecting $j$. Without loss of generality, we assume their indexes are $l_1,l_2,\cdots,l_{p(j)}$. Therefore, with the definition of adjacency matrix, the left of (\ref{n1}) can be represented as:
\begin{align*}
&\sum\limits_{i=1}^n\frac{1}{k(o_j)}\sum\limits_{l=1}^m\frac{a_{il}a_{jl}}{k(u_l)}
=\frac{1}{k(o_j)}\sum\limits_{i=1}^n\sum\limits_{l=1}^m\frac{a_{il}a_{jl}}{k(u_l)}\\
=&\frac{1}{p(j)}\sum\limits_{i=1}^n\bigg(\frac{a_{i,l_1}}{k(u_{l_1})}+\frac{a_{i,l_2}}{k(u_{l_2})}
+\cdots+\frac{a_{i,l_{p(j)}}}{k(u_{l_{p(j)}})}\bigg)\\
=&\frac{1}{p(j)}\bigg[\sum\limits_{i=1}^n\frac{a_{i,l_1}}{k(u_{l_1})}+\sum\limits_{i=1}^n\frac{a_{i,l_2}}{k(u_{l_2})}
+\cdots+\sum\limits_{i=1}^n\frac{a_{i,l_{p(j)}}}{k(u_{l_{p(j)}})}\bigg]\\
=&\frac{1}{p(j)}(\underbrace{1+1+\cdots+1}_{p(j)})=1
\end{align*}
The proof is completed.
\end{proof}

\begin{rem}
In fact, it is easy to obtain the above property from the point of matrix. At first, from the definition of $w_{ij}$, we have
\begin{eqnarray}\label{u}
W=(AU^{-1})(O^{-1}A)^T,
\end{eqnarray}
where $W=(w_{ij})\in R^{n\times n}$, diagonal matrix $U=\mathrm{diag}\{k(u_1),\cdots,k(u_m)\}\in R^{m\times m}$ means the degree matrix in $V_2$, so $AU^{-1}$ is a matrix with its each column sum is $1$, i.e.,
\begin{align*}
(1,\cdots,1)_{1\times n}\cdot(AU^{-1})=(1,\cdots,1)_{1\times m},
\end{align*}
Similarly, the diagonal matrix $O=\mathrm{diag}\{k(o_1),\cdots,k(o_n)\}\in R^{n\times n}$ means the degree matrix in set $V_1$, so $O^{-1}A$ is a matrix with its each row column is $1$, i.e., $(O^{-1}A)^T$ is also a matrix with its each column sum is $1$. Therefore,
\begin{align*}
(1,\cdots,1)_{1\times m}\cdot(O^{-1}A)^T=(1,\cdots,1)_{1\times n}.
\end{align*}
Thus,
\begin{eqnarray*}
(1,\cdots,1)_{1\times n}\cdot W=(1,\cdots,1)_{1\times n}\cdot(AU^{-1})(O^{-1}A)^T=(1,\cdots,1)_{1\times n},
\end{eqnarray*}
That is to say, matrix $W$ is also a matrix with its each column sum is $1$. Furthermore, one can also get that matrix $W$ has a left eigenvector $(1,\cdots,1)^T\in R^{n\times 1}$ corresponding to eigenvalue $1$.

Moreover, in the form of matrix, we can also easily get one right eigenvector for matrix $W$ is $(k(o_1),\cdots,k(o_n))^T\in R^{n\times 1}$ corresponding to eigenvalue $1$, since
\begin{align*}
&W\cdot(k(o_1),\cdots,k(o_n))^T\\
=&(AU^{-1})(O^{-1}A)^T\cdot(k(o_1),\cdots,k(o_n))^T=(AU^{-1})A^TO^{-1}\cdot(k(o_1),\cdots,k(o_n))^T\\
=&(AU^{-1})A^T\cdot(1,\cdots,1)^T_{n\times 1}=AU^{-1}(k(u_1),\cdots,k(u_m))^T\\
=&A(1,\cdots,1)^T_{m\times 1}=(k(o_1),\cdots,k(o_n))^T
\end{align*}

\end{rem}
For example, in Figure \ref{f1}, the adjacency matrix is
$A=\left(\begin{array}{cccc}1&1&0&1\\0&1&1&0\\0&1&1&1\end{array}\right)$, thus
\begin{align}
W&=(AU^{-1})(O^{-1}A)^T\nonumber\\
&=A\cdot \left(\begin{array}{cccc}1&0&0&0\\0&1/3&0&0\\0&0&1/2&0\\0&0&0&1/2\end{array}\right)
\cdot A^T \cdot
\left(\begin{array}{ccc}1/3&0&0\\0&1/2&0\\0&0&1/3\end{array}\right)
\nonumber\\
&= \left(\begin{array}{cccc}1&1/3&0&1/2\\0&1/3&1/2&0\\0&1/3&1/2&1/2\end{array}\right)
\cdot
\left(\begin{array}{ccc}1/3&0&0\\1/3&1/2&1/3\\0&1/2&1/3\\1/3&0&1/3\end{array}\right)
\nonumber\\
&= \left(\begin{array}{ccc}11/18&1/6&5/18\\1/9&5/12&5/18\\5/18&5/12&4/9\end{array}\right)
\label{matrix}
\end{align}

Therefore, we obtain the matrix $W$, which can also be found in (\ref{p}), all the elements are nonnegative, and its each column sum is $1$. Moreover, $(1,1,1)^T$ is the left eigenvector for $W$ corresponding to the eigenvalue $1$, while $(3,2,3)^T$ is the right eigenvector for $W$ corresponding to the eigenvalue $1$.

Next, we will discuss some properties of this matrix $W$. In order to do that, some useful lemmas should be introduced.
\begin{lem}\label{g}
(See \cite{Horn}) (Gersgorin Disc Theorem)
Let $A=(a_{ij})\in R^{n\times n}$ and
\begin{align*}
C_j^{\prime}(A)=\sum\limits_{i=1,i\ne j}^n|a_{ij}|, \quad~~~~ 1\le j\le n
\end{align*}
denote the deleted absolute column sums of $A$. Then all eigenvalues of $A$ are located in the union of $n$ discs
\begin{eqnarray*}
\bigcup\limits_{j=1}^n\{z\in \mathcal{C}: |z-a_{jj}|\le C_j^{\prime}\}.
\end{eqnarray*}
\end{lem}

\begin{lem}\label{pf}
(See \cite{Horn}) (Perron-Frobenius Theorem)
For $n\times n$ matrix $A$, which is irreducible nonnegative matrix, then
\begin{enumerate}
  \item $\rho(A)>0$, where $\rho(A)$ is the spectral radius;
  \item $\rho(A)$ is an eigenvalue of $A$;
  \item there is an vector $x>0$ and $Ax=\rho(A) x$;
  \item $\rho(A)$ is an algebraically (and hence geometrically) simple eigenvalue of $A$.
\end{enumerate}
\end{lem}

\begin{ass}
In the following discussion, we will always assume graph $G$ is connected.
\end{ass}

In fact, this assumption is not very strong. Because if $G$ is not connected, then $G$ can be split into at least two connected components, for example $G_1$ and $G_2$, then $V_1=V_1^1\bigoplus V_1^2$ and $V_2=V_2^1\bigoplus V_2^2$, such that $(V_1^1,V_2^1)$ and $(V_1^2,V_2^2)$ are two bipartites with no connection, and this case contradicts to our aim of using resource location feedback. Moreover, under this condition, one can also know that matrix $W$ is irreducible.

\section{Main results}
Now, it is time to bridge the gap between GRM and NBI. Recalling the process of NBI, one can find that the two steps $V_1\rightarrow V_2$ and $V_2\rightarrow V_1$ can be regarded as a round of resource diffusion, and the obtained resource vector $F^{(1)}$ is determined by the original resource vector $F$ and the matrix $W$, see (\ref{m}). Inspired by NBI, we propose a {\bf multiple rounds of resource diffusion algorithm}, which can be described as
\begin{align}\label{M-NBI}
\begin{array}{lllll}
V_1\rightarrow V_2\rightarrow &V_1\rightarrow V_2\rightarrow &V_1\rightarrow V_2\rightarrow &V_1\rightarrow V_2\rightarrow
&V_1\rightarrow\cdots\\
|~~~\mathrm{1st ~~round}&|~~~\mathrm{2nd ~~round}&|~~~\mathrm{3rd ~~round}&|~~~\mathrm{4th ~~round}&|~~~~~~~~~\cdots\\
F&F^{(1)}&F^{(2)}&F^{(3)}&F^{(4)}~~~\cdots
\end{array}
\end{align}
After $N$ round, the resource location $F^{(N)}$ is determined as:
\begin{align}
F^{(N)}=WF^{(N-1)}=W^2F^{(N-2)}=\cdots=W^{N-1}F^{(1)}=W^NF.
\end{align}

In the next, we will explore the property of matrix $W$ and $W^N$, in order to finalize the relationship between NBI and GRM.

\begin{thm}
For the matrix $W\in R^{n\times n}$ defined in (\ref{u}), suppose it it irreducible, then it has the following properties:
\begin{enumerate}
  \item $M$ must have an eigenvalue $\lambda_1=1$, and its multiplicity is $1$. Suppose $e_l$ and $e_r$ are the corresponding left eigenvector and right eigenvector, and $e_l^Te_r=1$.
  \item $e_{r}>0$, and $e_r$ can be chosen as:
\begin{eqnarray}\label{right}
e_r=(k(o_1),k(o_2),\cdots,k(o_n))^T£¬
\end{eqnarray}
while
\begin{eqnarray}\label{left}
e_l=\alpha (1,\cdots,1)^T, \alpha=1/\sum_{j=1}^nk(o_j).
\end{eqnarray}
  \item The other $n-1$ eigenvalues $\lambda_j\in \mathcal{C}, j=2,3,\cdots,n$ satisfying $|\lambda_j|<1$.
  \item $\lim\limits_{N\rightarrow+\infty}W^N=e_re_l^T$.
\end{enumerate}
\end{thm}
\begin{proof}
According to the Gersgorin disc theorem (Lemma \ref{g}), Perron-Frobenius theorem (Lemma \ref{pf}), Property 1 and Remark 1, one can easily get the conclusions 1, 2 and 3. Next, we will concentrate on proving the fourth point.

Denote $J$ as the Jordan form of $W$, i.e.,
\begin{align}
W=PJP^{-1},
\end{align}
where $J$ can be written in the form as: $J=\mathrm{diag}\{1,J_2,\cdots,J_q\}$, where
\begin{align*}
J_{\theta}=\left(\begin{array}{ccccc}
\lambda&1&&&\\&\lambda&1&&\\
&&\ddots&\ddots&\\&&&\lambda&1\\
&&&&\lambda\end{array}\right), \theta=2,\cdots,q
\end{align*}
where $\lambda$ are chosen from $\lambda_i, i=2,\cdots,n$, so $|Re(\lambda_i)|<1$, and
\begin{align*}
J_{\theta}^N\rightarrow 0, N\rightarrow +\infty.
\end{align*}
Therefore, if we denote $W^{\star}=\lim\limits_{N\rightarrow +\infty}W^N$, then
\begin{eqnarray*}
W^N=(PJP^{-1})^N=PJ^NP^{-1}\rightarrow P\mathrm{diag}(1,0,\cdots,0)P^{-1}, N\rightarrow +\infty,
\end{eqnarray*}
i.e., $W^{\star}=P\mathrm{diag}(1,0,\cdots,0)P^{-1}$.

Because $AP=PJ$, so the first column of $P$ is $e_r$. Similarly, because $P^{-1}A=JP^{-1}$, so the first row of $P^{-1}$ is $e_l^T$. Since $P^{-1}P=I$, therefore, $e_l^Te_r=1$, which is satisfied in the first point.

In all, we can get that $W^{\star}=e_re_l^T$. The proof is completed.
\end{proof}

For example, for the matrix $W$ in (\ref{matrix}), simple calculations can show that eigenvalues of $W$ are: $\lambda_1=1, \lambda_2=0.4034, \lambda_3=0.0689$; while $e_r=(3, 2, 3)^T$ and $e_l=1/8\cdot (1, 1, 1)^T$ for eigenvalue $1$, so:
\begin{eqnarray}\label{w2}
W^{\star}=e_re_l^T=
\left(\begin{array}{c}3\\2\\3\end{array}\right)\cdot 1/8\cdot (1,1,1)
=\left(\begin{array}{ccc}3/8&3/8&3/8\\2/8&2/8&2/8\\3/8&3/8&3/8\end{array}\right)
\end{eqnarray}
Next, we use the Matlab to verify the correctness of our claim.
\begin{eqnarray*}
W^{2}
=\left(\begin{array}{ccc}
0.4691  &  0.2870  &  0.3395\\
0.1914  &  0.3079  &  0.2701\\
0.3395  &  0.4051  &  0.3904
\end{array}\right);
W^{3}
=\left(\begin{array}{ccc}
    0.4129  &  0.3392  &  0.3609\\
    0.2262  &  0.2727  &  0.2587\\
    0.3609  &  0.3881  &  0.3804
\end{array}\right);\\
W^{4}
=\left(\begin{array}{ccc}
    0.3903  &  0.3606  &  0.3693\\
    0.2404  &  0.2591  &  0.2536\\
    0.3693  &  0.3803  &  0.3771
\end{array}\right);
W^{5}
=\left(\begin{array}{ccc}
    0.3812  &  0.3692 &   0.3727\\
    0.2461  &  0.2537 &   0.2514\\
    0.3727  &  0.3772 &   0.3758
\end{array}\right);\\
W^{6}
=\left(\begin{array}{ccc}
    0.3775 &   0.3727  &  0.3741\\
    0.2484 &   0.2515  &  0.2506\\
    0.3741 &   0.3759  &  0.3753
\end{array}\right);
W^{7}
=\left(\begin{array}{ccc}
    0.3760  &  0.3741  &  0.3746\\
    0.2494  &  0.2506  &  0.2502\\
    0.3746  &  0.3754  &  0.3751
\end{array}\right);\\
W^{8}
=\left(\begin{array}{ccc}
    0.3754  &  0.3746 &   0.3749\\
    0.2497  &  0.2502 &   0.2501\\
    0.3749  &  0.3751 &   0.3751
\end{array}\right);
W^{9}
=\left(\begin{array}{ccc}
    0.3752  &  0.3748 &   0.3749\\
    0.2499  &  0.2501 &   0.2500\\
    0.3749  &  0.3751 &   0.3750
\end{array}\right);\\
W^{10}
=\left(\begin{array}{ccc}
    0.3751  &  0.3749 &   0.3750\\
    0.2500  &  0.2500 &   0.2500\\
    0.3750  &  0.3750 &   0.3750
\end{array}\right);
W^{11}
=\left(\begin{array}{ccc}
    0.3750  &  0.3750 &   0.3750\\
    0.2500  &  0.2500 &   0.2500\\
    0.3750  &  0.3750 &   0.3750
\end{array}\right).
\end{eqnarray*}
Obviously, $W^{N}$ can converge to $W^{\star}$ in (\ref{w2}) when $N=11$, therefore, our claim is correct.

Based on the above theorem, considering the limit case $W^{\star}$, for any user $l$, since $f^{\prime}(o_i)=\sum_{j=1}^nw_{ij}^{\star}f(o_j)=\sum_{j=1}^nw_{ij}^{\star}a_{jl}$, therefore, we have
\begin{align}
\left(\begin{array}{c}f^{\prime}(o_1)\\f^{\prime}(o_2)\\ \cdots\\f^{\prime}(o_n)\end{array}\right)=W^{\star}\cdot\left(\begin{array}{c}a_{1l}\\a_{2l}\\ \cdots\\a_{nl}\end{array}\right)=e_re_l^T\cdot\left(\begin{array}{c}a_{1l}\\a_{2l}\\ \cdots\\a_{nl}\end{array}\right)=\alpha k(y_l)e_r=\alpha k(y_l)\left(\begin{array}{c}k(o_1)\\k(o_2)\\ \cdots\\k(o_n)\end{array}\right)
\end{align}
i.e., for recommendation one only needs to see the value of $e_r$, while $e_r$ denotes the vector composed of degree, and this is just the GRM.

\begin{rem}
The use of left eigenvector and right eigenvector in recommendation systems can be retrieved to \cite{Zhang07}. Moreover, it is also widely adopted in the analysis of synchronization and consensus literature, see \cite{Liu}.
\end{rem}

\begin{rem}
In fact, using the Hamilton-Cayley Theorem, matrix $W$ satisfies the polynomial $f(\lambda)=|\lambda I-W|$ with degree $n$, i.e., ${W}^{n}+a_1{W}^{n-1}+\cdots+a_{n-1}W+a_nI=0$. Moreover, according to the so-called minimal polynomial for matrix $W$, there exists a polynomial with degree $n^{\prime}\le n$, such that $P(W)={W}^{n^{\prime}}+a_1{W}^{n^{\prime}-1}+\cdots+a_{n^{\prime}-1}W+a_{n^{\prime}}I=0$.
Therefore, for any integer $N$, from the theory of polynomial division, there exist a quotient $Q(W)$ and a remainder $R(W)$, such that $W^N=Q(W)P(W)+R(W)$, where the degree of the remainder is less than the degree of the divisor $P(W)$, that is to say, any $W^N, N=1,2,\cdots$ can be equally described by a polynomial with degree less than $n^{\prime}$. In \cite{Zhou09}, the authors consider the second and even the third round, and simulations show that they exhibit better recommendation result than just only one round. Since we just want to emphasis the relationship between NBI and GRM from the theoretical view, here we omit the numerical simulations. Interested readers are encouraged to investigate this problem.
\end{rem}
\section{Conclusion}\label{conc}
In this paper, we first extend the NBI which can be regarded the first round of resource diffusion to the multiple rounds of resource diffusion algorithm. Then by rigorous theoretical analysis, we finally prove that the GRM is just the limit of our proposed algorithm. That is to say, for the multiple rounds of resource diffusion algorithm, NBI is the case with $N=1$, where $N$ means the times of diffusion, while GRM is the case with $N=+\infty$. We bridge the gap between GRM and NBI successfully.

\section*{Acknowledgment}
This work was supported by
the National Science Foundation of China under Grant No. 61203149, 61233016,
the National Basic Research Program of China (973 Program) under
Grant No. 2010CB328101, ``Chen Guang'' project supported by Shanghai
Municipal Education Commission and Shanghai Education Development
Foundation under Grant No. 11CG22, the Fundamental Research Funds
for the Central Universities under Grant No. 20140764.

\end{document}